\documentclass[aps,pra,twocolumn,groupedaddress]{revtex4}

\usepackage[utf8]{inputenc}
\usepackage[T1]{fontenc}
\usepackage[english]{babel}
\usepackage{graphicx}
\usepackage{algorithm2e}
\usepackage{hyperref}

\usepackage{enumerate,amsthm,amsmath,amssymb,color,graphicx,bbm,float,framed,amsfonts,verbatim}
\usepackage{natbib}

\newtheorem{theorem}{Theorem} 
\newtheorem{lemma}[theorem]{Lemma}

\newtheorem{definition}[theorem]{Definition}

\usepackage{epstopdf,hyperref}

\newcommand{\be}{\begin{equation}}
\newcommand{\ee}{\end{equation}}

\newcommand{\eps}{\varepsilon}

\newcommand{\Cliff}{\mathrm{CC}}
\newcommand{\Tree}{\mathrm{Tree}}

\begin{document}

\title{Practical Position-Based Quantum Cryptography}

\author{Kaushik Chakraborty\footnote{Inria, EPI SECRET, B.P. 105, 78153 Le Chesnay Cedex, France. Email: \texttt{kaushik.chakraborty@inria.fr}.},\quad Anthony Leverrier\footnote{Inria, EPI SECRET, B.P. 105, 78153 Le Chesnay Cedex, France. Email: \texttt{anthony.leverrier@inria.fr}.},}

\affiliation{Inria, EPI SECRET, B.P. 105, 78153 Le Chesnay Cedex, France}

\begin{abstract}
We study a general family of quantum protocols for position verification and present a new class of attacks based on the Clifford hierarchy. These attacks outperform current strategies based on port-based teleportation for a large class of practical protocols. 
We then introduce the Interleaved Product protocol, a new scheme for position verification involving only the preparation and measurement of single-qubit states for which the best available attacks have a complexity exponential in the number of classical bits transmitted.
\end{abstract}

\maketitle

\section{Introduction}

The goal of position-based cryptography is for an honest party to use her spatio-temporal position as her only credentials in a cryptographic protocol. In particular, 
Position verification aims at verifying that a certain party, called the prover, holds a given position in space-time. Such a protocol typically goes as follows: a set of 
verifiers will coordinate and send some challenge to the prover, and it is expected that only someone sitting in the supposed position of the prover can successfully pass the 
challenge. 

Position verification protocols have been studied in the classical setting where the challenges are described by classical information, and it was shown in \cite{CGM09} that 
information-theoretic security could never be obtained in the standard (Vanilla) model. More precisely, it is always possible for a \emph{coalition} of adversaries to convince 
the verifiers, even if none of the adversaries sits in the spatio-temporal region where the prover is supposed to be. Note, however, that the same paper gives secure constructions 
in the Bounded-Retrieval Model, which is a variant of the Bounded-Storage Model \cite{mau92}.
A possible way-out of this no-go theorem would be to consider a quantum setting. Indeed, several classical tasks which are known to be impossible in the classical domain can be 
achieved in the quantum domain: this is the case for instance of secret key expansion \cite{SBC09}, randomness amplification \cite{CR12} or randomness expansion \cite{VV12}. 

Position-based cryptography in the quantum setting was first investigated under the name of \emph{quantum tagging} by Kent around 2002, but only appeared in the literature much 
later in \cite{KMS11} where attacks against possible quantum constructions are described. Malaney independently introduced a quantum position verification scheme in \cite{Mal10}.
An example of a quantum protocol for position verification is one with two verifiers: one sending a qubit $|\phi\rangle = U|x\rangle$ with $x \in \{0,1\}$ and $U$ some unitary, and 
the second verifier sending a classical description of the unitary $U$. The task for the prover is then to measure the qubit in the basis $\{ U|0\rangle, U|1\rangle\}$ and to return the classical value of $x$ to both provers. 
There are many variations around this protocol, and the intuition for the possible security of such protocols is that only someone sitting in $P$ can obtain both $U$ and 
$|\phi\rangle$, perform the required measurement, and return the correct value $x$ on time. 
In \cite{LL11}, Lau and Lo extended the attack from \cite{KMS11} to show that the above intuition is incorrect if the unitary $U$ is a \emph{Clifford} gate. In that case, 
a couple of cheaters, Alice lying between $V_0$ and $P$, and Bob lying between $V_1$ and $P$, can always fool the verifiers provided that they share a small number of EPR pairs. 
This result was later generalized by Buhrman \textit{et al.} \cite{BCF11} who showed that such an attack always exists provided that the coalition of cheaters share sufficiently many EPR pairs: no position-based quantum cryptographic protocol can display information-theoretic security.

Two general families of attacks against such position-verification protocols have been considered in the literature so far, both based on quantum teleportation. The first one is 
inspired by Vaidman's protocol for nonlocal computation \cite{Vai03} and consists in the cheaters teleporting some quantum state back and forth, with the number of exchanges 
depending on the success probability of the attack. If the position-based protocol involves $n$ qubits, the resource (number of EPR pairs) required for this type of attacks to 
succeed typically scales double-exponentially with $n$ \cite{BCF11}.
Another class of attacks uses \emph{port-based} teleportation \cite{IH08} and requires only exponential entanglement to succeed \cite{BK11}. If one could prove that such an 
attack was indeed optimal, one would obtain a secure position-based protocol for all practical purposes. 

A different class of position-based verification protocols based on the nonlocal computation of Boolean functions was introduced by Buhrman \textit{et al.} in \cite{BFS13}, for 
which they suggested a new type of attacks based on the \emph{Garden-hose complexity} of the Boolean function. They showed in particular that finding an explicit Boolean function 
with polynomial circuit complexity (so that the honest prover can compute it) but exponential attack complexity in the garden-hose model is at least as difficult as separating the 
classes of languages P and L, corresponding respectively to decision problems decidable in polynomial time or logarithmic space. 
This result was recently extended by Klauck and Podder who showed that explicit Boolean functions on $k$ variables with Garden-hose complexity $\Omega(k^{2 + \eps})$ will be 
hard to obtain \cite{KP14}. These results give us little hope of finding an explicit position-verification based on the nonlocal computation of Boolean functions both practical and secure. 

Establishing lower bounds for the amount of entanglement shared by the coalition in order to successfully attack the protocol is a non trivial task. Current lower bounds are linear in the security parameter of the protocol  \cite{BK11}, \cite{TFK13}. 
Recently, a tight (linear) lower bound was proved for the BB84-based protocol where the unitary $U$ is either the identity or a Hadamard gate, in a model where the cheaters share an initial entangled state but are not allowed to exchange quantum communication during the protocol \cite{RG15}.
It was also shown by Unruh that security of some position-verification protocols could be established in the quantum random oracle model, that is if one has access to one-way functions \cite{Unr14}.

Recently, Qi and Siopsis initiated the study of imperfections in quantum position-based schemes, in particular in the presence of losses in the quantum channel between the 
verifiers and the prover \cite{QS15}. Indeed, in order to achieve practical distances between the verifiers and the prover it is necessary for the the protocol to be reasonably 
loss-tolerant. 

In this paper, we investigate the family of protocols described above, where the state $|\phi\rangle$ and the unitary $U$ is chosen from a family of $n$-qubit gates. We present 
some new attacks against such protocols that might become particularly efficient when the position-verification protocol is practical for the honest prover. 
We then introduce a new practical position-verification scheme involving only single-qubit operations, for which the best known attacks require an exponential amount of entanglement.

\section{A general family of position-verification protocols}

For simplicity, we mainly focus on one-dimensional protocols where two verifiers $V_0$ and $V_1$ aim at verifying the position of a prover $P$ located between them. We note that 
complications occur when dealing with more realistic 2 or 3-dimensional protocols (see for instance \cite{Unr14}), but explicitly avoid these questions here. 
Moreover, without loss of generality, we can always assume that the position $P$  of the prover is exactly at equal distance to $V_0$ and $V_1$ and that it takes one unit of time for light to travel from $V_0$ (or $V_1$) to $P$.  

Roughly speaking, a general position-verification protocol consists of three distinct phases:
\begin{itemize}
\item \emph{the preparation phase}, where $V_0$ and $V_1$ prepare a challenge for the prover. The challenge typically involves a quantum state (for instance an $n$-qubit state, 
or $n$ single-qubit states in the protocols considered in the present paper) as well as some classical information. The challenge is always given to the prover in a distributed 
fashion, one part coming from $V_0$, the other part coming from $V_1$.

\item \emph{the execution phase}, during which $V_0$ and $V_1$ send their respective share of the challenge towards the prover $P$, who solves the challenge she is given, and 
returns her answer to the verifiers. 

\item \emph{the verification phase}, during which the verifiers check that $(i)$ the answer is correct, and that $(ii)$ they received it not more than two time units after the beginning of the protocol. This assumes the idealized scenario where all communications are performed at the speed of light, and local computation take negligible time. 
Even in that idealized scenario, it makes sense to allow the honest prover to err a small fraction of the time. For this reason, the provers accept the answer if it meets 
some tolerance threshold $\eta$. In fact, one should distinguish between two sources of imperfections, losses and noises, and the tolerance threshold should therefore specify 
the amount of losses (i.e.~no answer from the prover) and noise (i.e.~incorrect answer) that can be tolerated. 
\end{itemize}

In this paper, we will first focus on an important family of position verification protocols where $V_0$ sends an $n$-qubit state and $V_1$ sends the classical description of a measurement basis, 
and the prover is required to measure the state in the correct measurement basis and to communicate the outcome to both verifiers. These protocols have been widely discussed in the 
literature for instance in \cite{KMS11} or \cite{LL11}.
In Section \ref{IP}, we will then introduce the Interleaved Product protocol where the description of measurement basis is transmitted to the prover as a product of a large number of single-qubit unitaries $\prod_u u_i v_i$, where the unitaries $\{u_i\}$ and $\{v_i\}$ are respectively described to the prover by $V_0$ and $V_1$.  
This scheme appears to be reasonably new, although similar ideas, with more verifiers, were already considered in \cite{LL11}. We note that the \emph{interleaved group product} 
(i.e.~$\prod u_i v_i$ where the $\{u_i\}$ and $\{v_i\} )$ are described by different verifiers) has been considered in the communication complexity literature, for instance in a 
recent paper by Gowers and Viola \cite{GV15}.

Before defining these protocols more formally, let us comment on some assumptions we make here. 
In this paper, our main goal is to present some natural position verification protocols and to study general classes of attacks that can be carried out by coalitions of cheaters. While we try to 
be as general as possible, we think it is sensible to make some specific choices in order to simplify the analysis. For instance, we restrict our protocols to using qubit states, 
and more importantly, we consider one-dimensional protocols with only 2 verifiers. 
Most of our analysis would carry through to arbitrary qudit protocols involving many verifiers. 
We also decided to leave aside all the problems related to timing in order to focus on the genuinely quantum part of the procedure. This means that we consider that all 
communication (classical or quantum) is performed at the speed of light, and that all computation is instantaneous. These are obviously unrealistic assumptions, but dealing 
with more realistic ones can be done independently as the analysis we provide here (see for instance the work of Kent \cite{ken12}).
The main source of imperfection in a position verification protocol is the quantum channel between the verifiers and the prover, which can never be assumed to be perfect. In general, the 
channel is both lossy and noisy, which is why even an ideal prover cannot possibly pass the test perfectly. On the other hand, it makes sense to assume that the classical 
channels are essentially perfect (lossless and noiseless).

\subsection{Formal description of the position-verification protocols}

Following the literature, we will find it useful to describe the protocol in terms of distributed collaborative games, where two players, named Alice and Bob, independently 
receive some query from some referee, are allowed a single round of (bipartite) communication and need to output some answer. 
In the honest prover case, Alice and Bob hold the same spatial position and the prover has access to both their inputs. In the cheating coalition case, Alice and Bob sit 
respectively between $P$ and $V_0$ or between $P$ and $V_1$ and are only allowed one simultaneous round of communication. The main result of \cite{BCF11} is that if Alice and Bob can win the game with arbitrarily many rounds of communication, then they can also win it with a single simultaneous round, provided that they are sufficiently entangled. 

The main family of protocols we will consider corresponding to games denoted by $G(n, \mathcal{U}, \eta)$ where $n$ refers to the number of qubits involved in the protocol, 
$\mathcal{U}$ is a set of $n$-qubit unitaries, and $\eta$ is the tolerance threshold.
We will also write $G(n,k, \eta)$ when the set $\mathcal{U}$ is a subset of $C_k$, the $k^{\mathrm{th}}$ level of the Clifford hierarchy (see the appendix for a formal definition of the Clifford Hierarchy).
The protocol $G(n, \mathcal{U}, \eta)$ consists of the following phases:

\vspace{0.15in}

\textbf{Preparation Phase:}

\begin{enumerate}
 \item The verifier $V_0$ chooses an $n$-qubit unitary operator $U \in_R \mathcal{U}$ and 
an $n$-bit string $x = (x_1, \ldots, x_n) \in_R \{0,1\}^n$. $V_0$ prepares $|\psi\rangle = U|x\rangle$, where $|x\rangle = \bigotimes_{i=1}^n |x_i\rangle$ is a computational 
basis state.

 \item $V_0$ sends $x$ and $U$ to $V_1$ through some secure authenticated classical channel.
\end{enumerate}

\vspace{0.15in}

\textbf{Execution Phase:}

\begin{enumerate}
 \item $V_0$ sends the $n$ qubit quantum state $|\psi\rangle$ to prover $P$ at time $0$. $V_1$ sends the unitary $U$ to $P$ at time $\tau=0$.
 
 \item The prover $P$ receives both $|\psi\rangle$ and $U$ at time $\tau=1$.
 
 \item After receiving $|\psi\rangle$ and $U$, the honest prover $P$ computes $U^{\dagger}|\psi\rangle$ and measures it in computational basis, obtaining some outcome 
 string $y$. $P$ then sends back $y$ to both $V_0$ and $V_1$. 
 
\end{enumerate}

\textbf{Verification Phase:}

\begin{enumerate}
\item The prover $P$ wins the game if $V_0$ and $V_1$ receive the same string $y$ at time $\tau=2$, and if the Hamming distance between $x$ and $y$ is less than $\eta n$: $d_H(x,y) \leq \eta n$.
\end{enumerate}

In the literature, this family is often considered in the single qubit case, for instance with $\mathcal{U} = \{ \mathrm{id}, H\}$ where $H$ is the Hadamard gate \cite{CGM09, BCF11, RG15}. Then it makes sense to repeat the protocol $n$ times in order to build some statistics. 

In our case, we aim at giving a more general picture of the possible attacks working against this scheme and consider $n$-qubit gates. For such protocols, we will show that there exists a trade-off between the complexity of the protocol for the honest prover and the resources needed to break the protocol for a coalition of cheaters.

\subsection{Attacks strategies against position verification protocols}
\label{general-strategy}

As was proved in \cite{BCF11}, there always exists a working attack strategy against any position verification protocol that allows a coalition of adversaries to perfectly impersonate the 
honest prover. 
In the case of the one-dimensional protocols considered in this paper, such a coalition consists without loss of generality of 2 players, Alice ($A$) and Bob ($B$), with 
Alice lying on the line between $V_0$ and $P$, and Bob lying between $V_1$ and $P$.

The attack strategies we will consider have the following structure:
\begin{enumerate}
\item Alice and Bob initially share a (possibly entangled) initial bipartite state $\rho_{AB}$ of dimension to be specified later. Typically, $\rho_{AB}$ consists of many EPR pairs. 
\item Alice intercepts the communication from $V_0$, namely a quantum register $\rho_{C}$ (where $C$ stands for challenge), as well as some classical information.
\item Bob intercepts the classical communication from $V_1$. 
\item Depending on the classical information they received, Alice and Bob perform respectively a quantum measurement on their respective registers, $AC$ and $B$.
\item They forward all the classical information as well as the outcomes of the measurement to their partner. 
\item Finally, upon receiving this information, they prepare and send their response to the verifiers.
\end{enumerate}
The main question of interest is to decide how the dimension of $\rho_{AB}$, and more particularly the entanglement of this state, scales with the parameters of the position verification protocol.

This scenario allows us to see the cheating procedure as a distributed task, or game, where Alice and Bob are asked questions (possibly consisting of a quantum state), are 
allowed a single round of communication and are required to output some specific answer. They win the game if they fool the verifiers.

We can interpret the family $G(n, \mathcal{U}, \eta)$ in these terms:

\begin{definition} The distributed game $G(n, \mathcal{U}, \eta)$ is defined as follows:
\begin{itemize}
\item \textbf{Input:} $|\psi\rangle = U |x\rangle$ for Alice, $U \in \mathcal{U}$ for Bob
\item \textbf{Output:} $a \in \{0,1\}^n$ for Alice, $b\in \{0,1\}^n$ for Bob
\item \textbf{Winning condition:} $a=b$ and $d_H(a,x) \leq \eta n$
\end{itemize}
\end{definition}

We now list a few questions of interest. 
In the perfect setting ($\eta=0$), how many EPR pairs do Alice and Bob need to share to carry out a successful attack with reasonable probability?
One of the main open questions of the field is to find an explicit protocol that requires an exponential number of EPR pairs to break.

Second, if $\eta>0$, this opens the door to new attacks, even for non entangled cheaters. A possible strategy consists in Alice measuring the state in a random basis and forwarding her measurement outcome to Bob.
Ideally, it would be interesting to understand how the amount of entanglement required for cheating behaves as a function of $\eta$.

We should also comment on the definition of a successful attack. If the goal is to design a secure protocol, then Alice and Bob should not be able to cheat, even with a very small probability. Indeed, even if the cheating strategy only succeeds with probability $10^{-2}$ or $10^{-3}$, it is difficult to claim that the protocol is secure. Ideally, we want this cheating probability to be exponentially small in $n$.
In this paper, however, we choose for simplicity to focus on attacks that work with high probability (close to 1).

\section{Attacks for $\eta=0$ based on the Clifford hierarchy}
\label{perfect}

In this section, we first study attack techniques based on the Clifford hierarchy that can be applied by cheaters against the family of protocols $G(n, \mathcal{U}, 0)$ in the case where the value of the tolerance threshold $\eta$ is set to 0. The definition of the Clifford hierarchy is given in the appendix. Let us simply recall here that the first 
two levels $C_1(n)$ and $C_2(n)$ of the hierarchy correspond respectively to the Pauli and the Clifford groups. 

In particular, we will give explicit attacks that may be efficient in the following practically relevant cases: $(1)$ if $\mathcal{U} \subseteq C_k(n)$, that is if the unitaries all belong to some low level $k$ of the Clifford hierarchy,
$(2)$ if the unitaries in $\mathcal{U}$ can all be implemented with a quantum circuit with a fixed layout. 

We note that these two cases correspond to protocols that appear to be practical for a honest prover. Indeed, gates in a low level of the Clifford Hierarchy are much easier 
to implement fault tolerantly than arbitrary gates. 
Moreover, if the quantum states are photonic states, and the honest prover uses integrated photonics to implement the unitaries in $\mathcal{U}$, a fairly reasonable choice 
in practice, then it makes sense to fix some layout, that is an optical circuit consisting of single or 2-qubit gates for instance, and to obtain the family $\mathcal{U}$ by 
changing the value of the single and 2-qubit gates.

\subsection{A general attack for $\mathcal{U} = C_k$}

Let us first define the \emph{Clifford complexity} of a family $\mathcal{U}$ of unitaries. 
 
 \begin{definition}
 Let $\mathcal{U}$ be a set of $n$-qubit unitaries. 
 We define the \emph{Clifford complexity} of the set $\mathcal{U}$, denoted by $\Cliff[\mathcal{U}]$, to be the minimum number of EPR pairs that Alice and Bob must share to perfectly win the game $G(n, \mathcal{U}, 0)$.
 \end{definition}

It is easy to see that if the unitary $U$ is a Pauli matrix, then Alice and Bob can win the game $G(n,k=1,0)$ without sharing any entanglement because $|\psi\rangle$ is also a basis state $|y\rangle$. The two strings $x$ and $y$ coincide on the qubits for which $U$ is the identity or a $Z$ Pauli matrix, and differ for the other qubits. Therefore, Alice simply needs to measure $|\psi\rangle$ in the computational basis and to forward her results to Bob, who can recover the correct string $x$ using his knowledge of $U$.
This shows that 
$$\Cliff[C_1(n)] = 0.$$

If the unitary $U$ belongs to the Clifford group $C_2$, then Alice and Bob can again win the game perfectly if they share $n$ EPR pairs. The idea is for Alice to teleport the state $|\psi\rangle$ to Bob using the $n$ EPR pairs. Bob obtains the state $\sigma |\psi\rangle$ where $\sigma \in C_1(n)$ is a Pauli correction. 
Applying the unitary $U^\dagger$ to his state, Bob obtains
$$ U^\dagger \sigma   |\psi\rangle = U^\dagger \sigma U |x\rangle,$$
where $ U^\dagger \sigma U \in C_1(n)$. This means that Bob simply needs to measure this state in the computational basis, and forward his result to Alice. 
Once they know both the value of $\sigma$ and the result of the measurement, both Alice and Bob are able to recover the correct value of the string $x$ and they win the game. 
This proves that 
$$\Cliff[C_2(n)] \leq n.$$

If the unitary $U$ to be implemented belongs to the $k^{\mathrm{th}}$ level of the Clifford hierarchy, then Alice and Bob can apply an iterative procedure which is described in Algorithm \ref{algo_cliff}. This algorithm is similar to the protocol of Vaidman \cite{Vai03} for instantaneously measuring nonlocal variables and to the cheating strategy of \cite{BCF11}. The main difference lies in the termination condition: here, the algorithm terminates after a deterministic number of rounds that depends on the considered level of the Clifford Hierarchy.

\begin{widetext}
\begin{tabular}{p{0.95\textwidth}}
\RestyleAlgo{boxed}

\begin{algorithm}[H]
\KwIn{ $|\psi\rangle = U|x\rangle$ received by Alice, $U = U_0 \in C_k$ received by Bob }
\KwOut{$x \in \{0,1\}^n$}
\BlankLine
\nl Alice teleports the state $|\psi\rangle$ to $B$ using $n$ EPR pairs and obtains a string describing $\sigma_{A_1} \in\mathcal{P}_n$. 
Bob obtains the state $\sigma_{A_1} |\psi\rangle=\sigma_{A_1} U |x\rangle$.\\
\BlankLine
 \nl Bob applies $U^\dagger$ to his state and teleports the outcome $U^\dagger \sigma_{A_1} U |x\rangle$ to Alice, obtaining some classical description of $\sigma_{B_1} \in \mathcal{P}_n$.
 Alice obtains the state $U_1 |x\rangle$ where $U_1 =\sigma_{B_1}U^\dagger \sigma_{A_1}U \in C_{k-1}$.\\
 \BlankLine
 \For{$j=1$ to $k-3$}{
 \BlankLine
\nl Alice knows the value of $\sigma_{A_1}, \ldots, \sigma_{A_j}$ (among the $4^{jn}$ possibilities). 
Alice and Bob share $4^{n}\times (n 4^{(j-1)n)})$ EPR pairs devoted to Round $j$, corresponding to $4^{n}$ sets of $n\times 4^{(j-1)n}$ EPR pairs, one set for each possible value of $\sigma_{A_j}$.
Alice teleports back each of the $4^{(j-1)n}$ $n$-qubit states (of the form $U_j|x\rangle$ for some unitary $U_j \in C_{k-j}(n)$) she received from Bob using the ``teleportation channel'' indexed by $\sigma_{A_j}$. 
In that teleportation channel, Bob obtains the state $\sigma_{A_{j+1}} U_j |x\rangle$, applies $U_j^\dagger$ to that state, before teleporting it back to Alice in the corresponding teleportation channel.
Alice receives $U_{j+1} |x\rangle$ with $U_{j+1} = \sigma_{B_{j+1}} U_j^\dagger A_{j+1} U_j \in C_{k-(j+1)}$.
}

\BlankLine
\nl Alice uses a final round of teleportation for the $4^{(k-2)n}$ $n$-qubit states, and obtains a classical description of $\sigma_{A_{k-1}}$. \\
\BlankLine
\nl Alice sends the classical value of $\sigma_{A_1}, \ldots, \sigma_{A_{k-1}}$ to Bob.\\
\BlankLine
\nl Bob applies $U_{k-1}^\dagger$ to each $n$-qubit state, measures in the computational basis, and forwards the classical output, as well as the value of $\sigma_{B_1}, \ldots, \sigma_{A_{k-2}}$ to Alice. \\
\BlankLine
\nl Both Alice and Bob compute the value of $x$.

 \caption{Cheating strategy for $G(n, C_k(n), 1)$ based on the Clifford hierarchy}
\label{algo_cliff}
\end{algorithm}
\end{tabular}

 \begin{figure}[h!]
\centering
 \includegraphics[scale = 0.73]{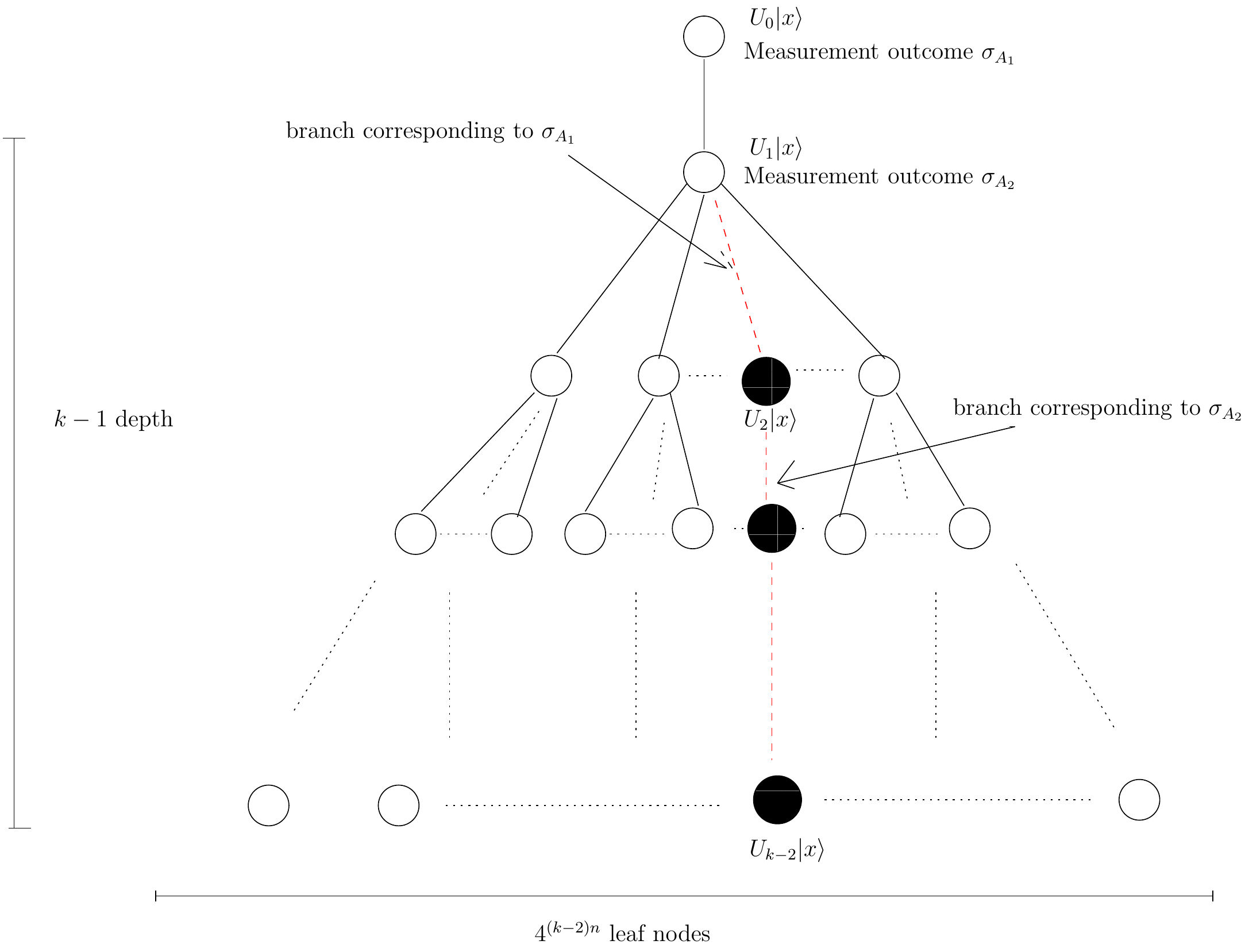}
\caption{Pictorial view of Step $3$ of Algorithm \ref{algo_cliff}: Each level of the tree corresponds to a round trip between Alice and Bob. Each of the nodes correspond to a quantum state. In particular, the root node is the initial quantum state $U_0 |x\rangle$ received by Alice, and the path in red dash (determined by the successive outputs of the Bell measurements) goes along the various states held by Alice at different steps of the protocol, namely $U_1|x\rangle, \ldots, U_{k-2}|x\rangle$.}

 \label{picture1}
\end{figure}

\end{widetext}

\begin{lemma}
If Alice and Bob apply Algorithm \ref{algo_cliff}, then they win the game.
\end{lemma}
\begin{proof}
To prove the correctness of the algorithm, we need to show that $U_j \in C_{k-j}$ and that Bob can perform $U_j^\dagger$ since he knows the value of $U_j$. 
The first point is shown by recurrence: $U_0 = U \in C_k$ and if $U_j \in C_{k-j}$, then $U_{j+1} = \sigma_{B_{j+1}} U_j^\dagger A_{j+1} U_j \in C_{k-j-1}$.
Moreover, the value of $U_j$ is a function of $U_{j-1}, \sigma_{A_j}$ and $\sigma_{B_j}$. For the quantum channel labeled by $\sigma_{A_j}$, Bob is therefore able to apply $U_j^\dagger$.
\end{proof}

The existence of the attack strategy described in Algorithm \ref{algo_cliff} allows us to obtain the following upper bound for the Clifford complexity of the set $C_k(n)$.
\begin{theorem}
 \label{thm1}
 \begin{align}
 \Cliff[C_k(n)] \leq  4n \, 4^{n(k-2)}.
 \end{align}
 \end{theorem}
 
  \begin{proof}
 The loop at Step $3$ in Algorithm \ref{algo_cliff} can be viewed as a branching tree with depth $k-2$ (see Fig.~ \ref{picture1}). This tree is regular with each internal node having $4^n$ children (corresponding to the $4^n$ possible values for Alice's Bell measurement result).
 Each layer of the tree corresponds to a round trip between Alice and Bob, that is $2n$ EPR pairs.
 Computing the complexity of the attack therefore amounts at counting the number of branches in the tree. For a tree of depth $k-2$, the number of branches is $\sum_{j=0}^{k-3} 4^{jn}$. Moreover, the last step of the protocol consists in a quantum teleportation of $n \times 4^{n(k-2)}$ qubits from Alice to Bob. 
 In total, the number of EPR pairs used in the protocols is therefore
 $$2n \sum_{j=0}^{k-2}4^{jn} + n 4^{n(k-2)} \leq 4n 4^{n(k-2)}.$$

 \end{proof}

In the following, we denote by $\Tree[C_k(n)]$ the number of EPR pairs required to perform the attack described by Algorithm \ref{algo_cliff} on the set of unitaries $C_k(n)$. Theorem \ref{thm1} simply says that
\begin{align}
 \Cliff[C_k(n)] \leq \Tree[C_k(n)] \leq 4n \, 4^{n(k-2)}.
\end{align}

\subsection{Attacks when $\mathcal{U}$ correspond to quantum circuits with a fixed layout}

 The attack corresponding to Algorithm \ref{algo_cliff} is general and works for any $n$-qubit gate in some given level of the Clifford hierarchy. 
 In the context of position verification protocols, however, the interesting set of gates $\mathcal{U}$ from which the unitary to be implemented is chosen, is often more restricted. Indeed, if the protocol is to be practical, then a honest prover should be able to implement the unitaries reasonably efficiently.
 For this reason, it is interesting to consider unitaries described by quantum circuits. 
 
 In a practical scenario, where the quantum states given to Alice are photonic qubits, it makes sense to consider photonic implementations for the quantum circuit, and therefore to consider unitaries with a fixed layout for the quantum circuit, and adjustable single and two-qubit gates. This is typically the case for experimental implementations based on integrated photonics \cite{OFV09}.
 
 For this reason, the set $\mathcal{U}$ of unitaries considered could be described by a fixed layout, and a specific unitary $U \in \mathcal{U}$ is then described by giving the value of each single or two-qubit gate in the layout. For a quantum circuit based on linear optics, the layout $\mathcal{L}$ corresponds to the position of the phase-shifters and beamsplitters, and the unitary is given by the specific values of the phase-shifts and transmission of the beamsplitters. 
 
 We will be interested in the complexity of attacks for such schemes as a function of the depth and width of such quantum circuits.
 
 \begin{definition}
 Let $\mathcal{L}$ be the layout for an $n$-qubit quantum circuit, consisting of adjustable elementary gates. The set  $\mathcal{U}_{\mathcal{L}}$ of $n$-qubit unitaries corresponds to the set of unitaries which can be implemented with a quantum circuit with layout $\mathcal{L}$.
 \end{definition}
 
 Let us prove elementary results about the composition of circuit layouts.

\begin{lemma}[Parallel circuits]
 \label{thm2}
Let $\mathcal{L}_1, \mathcal{L}_2$ be two layouts for quantum circuits. 
Then 
\begin{align}
\Cliff[\mathcal{U}_{\mathcal{L}_1} || \mathcal{U}_{\mathcal{L}_2} ]\leq  \Cliff[\mathcal{U}_{\mathcal{L}_1}]+ \Cliff[\mathcal{U}_{\mathcal{L}_2}],
\end{align}
where $\mathcal{L}_1 || \mathcal{L}_2$ is the layout corresponding to putting $\mathcal{L}_1$ and $\mathcal{L}_2$ in parallel. 
\end{lemma}

We note that the quantum unitary corresponding to two circuits in parallel is simply the tensor product of the unitaries: $U_{\mathcal{L}_1 || \mathcal{L}_2} = U_{\mathcal{L}_1} \otimes U_{\mathcal{L}_2}$ and therefore 
$$\mathcal{U}_{\mathcal{L}_1 || \mathcal{L}_2} \subset \mathcal{U}_{\mathcal{L}_1} \otimes \mathcal{U}_{\mathcal{L}_2}.$$

\begin{proof} 
Consider any gate $U_1 \otimes U_2 \in \mathcal{U}_{\mathcal{L}_1 || \mathcal{L}_2}$. Since both Alice and Bob know the decomposition $U_1 \otimes U_2$, they can implement the optimal attack for $U_1$ and for $U_2$ independently, since these unitaries act on distinct sets of qubits. The complexity of the overall attack is simply the sum of the complexities of implementing $U_1$ and $U_2$, which is upper bounded by $ \Cliff[\mathcal{U}_{\mathcal{L}_1}] + \Cliff[\mathcal{U}_{\mathcal{L}_2}]$.
\end{proof}

\begin{lemma}[Concatenated circuits]
 \label{thm3}
 Let $\mathcal{L}_1, \mathcal{L}_2$ be two layouts for quantum circuits. 
Then 
\begin{align}
\Cliff[\mathcal{U}_{\mathcal{L}_1\mathcal{L}_2 } ]\leq  \Tree[\mathcal{U}_{\mathcal{L}_1}]  \Tree[\mathcal{U}_{\mathcal{L}_2}],
\end{align}
where $\mathcal{L}_1  \mathcal{L}_2$ is the layout corresponding to concatenating the layouts $\mathcal{L}_1$ and $\mathcal{L}_2$.
\end{lemma}

\begin{proof}
The strategy consists in first applying the strategy corresponding to Algorithm \ref{algo_cliff} for unitary $U_1 \in \mathcal{U}_{\mathcal{L}_1}$. Then, at the last round, instead of measuring the state, Bob continues the teleportation protocol in order to implement $U_2 \in \mathcal{U}_{\mathcal{L}_2}$. There are at most $\Tree[\mathcal{U}_{\mathcal{L}_1}]$ nodes in the tree corresponding to the implementation of $U_1$, and it is sufficient to apply the protocol to each of the leaves in order to implement to concatenation of $U_1$ and $U_2$. Therefore,  $\Tree[\mathcal{U}_{\mathcal{L}_1}]  \Tree[\mathcal{U}_{\mathcal{L}_2}]$ EPR pairs are sufficient to implement the total unitary.
\end{proof}

From Lemmas \ref{thm2} and \ref{thm3}, it is possible to compute an upper bound for the Clifford complexity of any layout, as a function of its depth and size.

\begin{theorem}
 \label{circ_cliff}
 Let $\mathcal{L}$ be the layout of an $n$-qubit quantum circuit of depth $d$ where each layer consists of gates in $C_{k_i}$. 
 Then \begin{align}
 \Cliff[\mathcal{U}_{\mathcal{L}}] \leq 4^{n\sum_{i=1}^d (k_i-2)} \times (4n)^d.
 \end{align}
\end{theorem}

\begin{proof}
The layout $\mathcal{L}$ can be decomposed into $d$ layers: $\mathcal{L} = \mathcal{L}_1 \mathcal{L}_2 \cdots \mathcal{L}_d$.
By applying Lemma \ref{thm3} recursively, one obtains that
$$ \Cliff[\mathcal{U}_{\mathcal{L}}] \leq \prod_{i=1}^d \Tree[\mathcal{U}_{\mathcal{L}_i}].$$
Combining this with the result of Theorem \ref{thm1}, one finally obtains
$$ \Cliff[\mathcal{U}_{\mathcal{L}}] \leq \prod_{i=1}^d 4n 4^{n(k_i-2)},$$
which establishes the result.
\end{proof}

We note that this result can be slightly improved by using Lemma \ref{thm2} together with Theorem \ref{thm1} for the last layer. Indeed, if the last layer only consists of 1 or 2-qubit gates, then it can be implemented with at most $n \times (4n) \times 4^{2(k-2)}$ EPR pairs since the layer can be seen as at most $n$ parallel circuits acting on at most 2 qubits each.

We conclude this section with an important remark, which was already made in \cite{QS15}. If the value of $\eta$ is too large, then there always exists a winning strategy for non-entangled cheaters. For the protocols considered above, $\eta=1/2$ is always achievable by a simple random guessing strategy: Alice and Bob simply agree on a random string and return it to the verifiers. For specific protocols where the family $\mathcal{U}$ displays some structure, better attacks are available. For instance, in the case of the BB84 scheme, measuring in the Breidbart basis allows the cheaters to win if $\eta \geq 1-\cos^2(\pi/8) \approx 0.15$.

\section{The Interleaved Product protocol}
\label{IP}

In this section, we introduce a new scheme for position verification based on the interleaved group product. This scheme depends on two main parameters: the number $n$ of single-qubit states used and a parameter $t$ quantifying the size of the product. More formally, the \emph{Interleaved Product} protocol denoted by $G_\mathrm{IP}(n, t, \eta_{\mathrm{err}}, \eta_{\mathrm{loss}})$, goes as follows:

\textbf{Preparation Phase:}

\begin{enumerate}

 \item $V_0$ chooses a random bit string $x \in_R \{0,1\}^n$ and and a single-qubit unitary $U$ chosen from the Haar measure on unitary group $U(2)$. $V_0$ also chooses $2t-1$ additional independent unitaries $u_1, \ldots, u_t, v_1, \ldots, v_{t-1}$ from the Haar measure on $U(2)$ and computes $v_t = u_t^\dagger v_{t-1}^\dagger \ldots v_1^\dagger u_1^\dagger U$, thus ensuring that $U = \prod_{i=1}^t u_i v_i$. Verifier $V_0$ then informs $V_1$ of these choices thanks to a secure classical channel.

 \item $V_0$ prepares the $n$-qubit state $|\psi\rangle = U^{\otimes n} |x\rangle$ , applying the same unitary $U$ to all the qubits of $|x\rangle$.
 
\end{enumerate}

\vspace{0.15in}

\textbf{Execution Phase:}

\begin{enumerate}
\item At time $\tau=0$, $V_0$ sends the state $|\psi\rangle$ as well as the classical description of $(u_1, \ldots, u_t)$ to the prover, and $V_1$ sends the classical description of $(v_1, \ldots, v_t)$ to $P$. 

\item At time $\tau=1$, the prover receives $|\psi\rangle$, computes $U = \prod_{i=1}^t u_i v_i$, applies $(U^\dagger)^{\otimes n}$ to $|\psi\rangle$ and measures the resulting state in the computational basis, obtaining some outcome $y \in \{\emptyset, 0,1\}^n$, which is sent to both $V_0$ and $V_1$. Here the symbol $\emptyset$ refers to an empty measurement result.

\end{enumerate}

\vspace{0.15in}

\textbf{Verification Phase:}

\begin{enumerate}
\item The prover $P$ wins the game if $V_0$ and $V_1$ both receive an identical string $y$ at time $\tau=2$, if the number of errors is less than $\eta_{\mathrm{err}} n$ and the number of empty results $\emptyset$ is less than $\eta_{\mathrm{loss}} n$.
\end{enumerate}

Interestingly for this protocol, the verifiers only need to prepare arbitrary single-qubit states and the honest prover is simply required to measure a qubit in a given basis, which is quite practical. 
We note that a similar family of protocols was considered in \cite{LL11}, but with more verifiers, which made the protocol less practical. 
Here we make the choice that the same unitary $U$ is applied to all the qubits. A variant of the protocol would be to send $n$ successive challenges to the prover, with $n$ different choices for the unitary.

The main feature of this protocol is that the value of the unitary $U$ that defines the measurement basis is described by a product $U = \prod_{i=1}^t u_i v_i$ which is communicated to the prover in a distributed fashion. Intuitively, if a coalition of cheaters tries to break the protocol, it seems that they need to follow a back-and-forth strategy to take care of each of the unitaries, one at the time. As we will see in the next section, this leads to attacks with a complexity exponential in the parameter $t$. 
On the other hand, the honest prover simply needs to compute the $2t$-fold product of $2\times 2$ matrices, which takes time linear in $t$.

In fact, for a practical implementation, each of the $2t$ unitaries should be described with a given (finite) level of accuracy, meaning that describing a unitary is done with a constant number of bits. We ignore this subtlety in the present paper.

\section{Attack strategies for the Interleaved-Product protocol}

By construction, the Interleaved-Product protocol is immune to the attacks based on the Clifford hierarchy: this is simply because all the gates are chosen from the Haar measure and therefore do not belong to any low level of the Clifford hierarchy. 
Moreover, the product structure enforces a large depth (of order $2t$ which can be taken as arbitrarily large in practice) for the quantum circuit.
Note that in the proposal of \cite{LL11}, neither of these conditions was enforced because $t$ corresponded to the number of verifiers (which should remain quite small for practical protocols) and all the gates belong to some low level of the Clifford hierarchy.

There exist, however, some attacks working in the regime $\eta_{\mathrm{err}}>0$, which we investigate now. Recall that we consider here the lossless scenario where the prover is required to give a bit value 0 or 1 for each qubit.
The first strategy uses port-based teleportation over $2t$ rounds.
The second strategy we will consider relies on the Solovay-Kitaev theorem for approximating arbitrary gates with gates in a low level of the Clifford hierarchy, for which the attack of Algorithm \ref{algo_cliff} can be applied.    
Both attacks lead to the same complexity and require $2^{O(t \log(t/\eta_{\mathrm{err}}))}$ EPR pairs.
Both strategies work in the lossless case $\eta_{\mathrm{loss}}=0$.

We end this section with a discussion of possible attack strategies for non-entangled cheaters, which works if $\eta_{\mathrm{err}} + \eta_{\mathrm{loss}}/4 \geq 1/4$.

\subsection{Attack based on Port-based teleportation}

The attack proceeds as follows:
\begin{itemize}
\item Alice applies the unitary $u_1^\dagger$ to each of her $n$ qubits and uses $m_1$ EPR pairs to teleport each qubit to Bob. This consumes a total of $M_1 = m_1 n$ EPR pairs. 
\item Bob applies the unitary $v_1^\dagger$ to all of his qubits, and uses $m_2$ EPR pairs to teleport each one back to Alice. This consumes a total of $M_2 = m_2 M_1$ EPR pairs. 
\item This process is repeated for $2t$ rounds, after which the unitary $U^\dagger$ has been applied to all the qubits. At each step, Alice or Bob uses $m_i$ EPR pairs to perform the port-based teleportation of a single qubit. 
\item At the last step, Bob measures each qubit in the computational basis, and both he and Alice exchange their measurement results. 
\end{itemize}

There are two quantities of interest to analyze the attacks: the total number of EPR pairs used by Alice and Bob, and the fidelity of the final state. Recall indeed that port-based teleportation is not perfect, and that the teleported state is only an approximation of the input state. 

The number $M$ of EPR pairs is given by:
\begin{align}
M &= M_1 + M_2 + \cdots + M_{2t-1}\\
&= n \left[m_1 + m_1 m_2 + \cdots + \prod_{i=1}^{2t-1} m_i \right].
\end{align}
The fidelity $F$ between the qubit after the $2t-1$ rounds of teleportation and the initial qubit is:
\begin{align}
F &\geq \prod_{i=1}^{2t-1} \left(1- \frac{4}{m_i}\right).
\end{align}

Choosing the slightly suboptimal strategy where all the $m_i$ are taken to be equal to a constant $m$ gives: $M= n m \frac{m^{2t-1}-1}{m-1} \approx n m^{2t-1}$ and $F = (1-4/m)^{2t-1}$, that is:
\begin{align}
M \approx n \left(\frac{8t}{\eta_{\mathrm{err}}}\right)^{2t-1},
\end{align}
where $\eta_{\mathrm{err}} = 1-F$ is assumed to be small.
This establishes the following result.

\begin{theorem}
Port-based teleportation provides an attack strategy against $G_{\mathrm{IP}}(n,t,\eta_{\mathrm{err}}, \eta_{\mathrm{loss}}=0)$ that requires $n\exp( O(t \log (t/\eta_{\mathrm{err}})))$ EPR pairs.
\end{theorem}

\subsection{Attack based on the Solovay-Kitaev approximation}

We now consider a different attack strategy based on the Solovay-Kitaev approximation, which guarantees that any single-qubit unitary can be approximated with accuracy $\eps$ by a sequence of unitaries taken from some fixed universal set of gates. 

\begin{theorem}[Solovay-Kitaev \cite{NC10}]
  \label{solv_kit1}
  If $\mathcal{G} \subseteq SU(d)$ is a universal family of gates (where $SU(d)$ is the group of unitary operators in a $d$-dimensional
Hilbert space), $\mathcal{G}$ is closed under inverse and $\mathcal{G}$ generates a dense subset of $SU(d)$,
then for any $U \in SU(d)$, $\eps > 0$, there exist $ g_1, g_2, \ldots , g_l \in \mathcal{G}$ such that $\| U - U_{g_1}U_{g_2} \ldots U_{g_l} \| \leq \eps$ and $l = O(\log^c\left(\frac{1}{\eps}\right))$, where
$c < 3$ is a positive constant.
\end{theorem}

Let us fix $\mathcal{G} = \{H,T\}$ where $H$ is the Hadamard operator and $T$ is the $\frac{\pi}{8}$ qubit gate, and note that this set lies in the third level $C_3$ of the Clifford hierarchy.
The Solovay-Kitaev theorem guarantees that for each unitary $U_i$ used in the game $G_{\mathrm{IP}}(n,t,\eta_{\mathrm{err}}, \eta_{\mathrm{loss}})$, there exists another unitary $U_i'$, obtained as a product of exactly $l$ gates from $\{H, T, \mathbbm{1}_2\}$ (where the identity is chosen so that the size $l$ can be chosen to be independent the unitary $U_i$). 
By decomposing their respective gates $u_i$ and $v_i$ into products of gates in $C_3$, Alice and Bob are able to implement the attack strategy of Algorithm \ref{algo_cliff}.

\begin{theorem}
 \label{solv_kit_att}
 There exists an attack strategy for $G_{\mathrm{IP}}(n,t,\eta_{\mathrm{err}}, \eta_{\mathrm{loss}}= 0)$ requiring $2^{8t\log^c\left(2t/\eta_{\mathrm{err}}\right)}n$ EPR pairs, where $c <3$.
\end{theorem}

\begin{proof}
 According to Solovay-Kitaev theorem, one can approximate each unitary $U_i$ used in the protocol by another unitary $U'_i$ such that  $\| U_i - U'_i\| \leq \frac{\eta_{\mathrm{err}}}{2t}$, using a sequence of $l = O(\log^c(2t/\eta_{\mathrm{err}}))$ gates. 
 Overall, the approximation quality is given by
 $$\left\|\prod_{i=1}^t U_i V_i - \prod_{i=1}^t U_i' V_i' \right\|\leq \eta_{\mathrm{err}}.$$
 The circuit to implement the gate $\prod_{i=1}^t U_i' V_i' $ has depth $2tl$ and uses only gates from $C_2$ or $C_3$. 
 According to Theorem \ref{circ_cliff}, the number $M$ of EPR pairs needed to perform the attack is 
 \begin{align}
 M = 2^{8tl} = 2^{8t\log^c\left(2t/\eta_{\mathrm{err}}\right)}.
 \end{align}
 Performing this attack for each of the $n$ qubits proves the theorem.

\end{proof}

This attack can in fact be improved by noting that the gates in $\mathcal{G} = \{H, T\}$ are semi-Clifford (see the appendix for a definition).
Recall that for a semi-Clifford unitary $U$, there are $2^n$ operators $\sigma \in \mathcal{P}_n$ such that $U \sigma U^\dagger \in \mathcal{P}_n$. This implies that for such gates, the tree described in Algorithm \ref{algo_cliff} can be taken to have degree $4^n - 2^n$. 
For $n=1$, as is the case here, this means that the complexity of approximating $\prod_{i=1}^t U_i V_i$ can be reduced to $2^{4lt}$ instead of $2^{8lt}$, leading to an overall quadratic improvement in the complexity of the attack.

\subsection{Attacks for a non-entangled coalition of cheaters}

A possible cheating strategy for non-entangled cheaters was considered in \cite{QS15} and goes as follows: Alice measures each qubit $|\psi_i\rangle$ of the incoming state in a random basis, obtains some measurement result corresponding to a qubit state $|\tilde{\psi}_i\rangle$ and communicates the classical description of $\tilde{\psi}_i$ to Bob.
When Alice and Bob learn the value of the unitary $U = \prod_{i=1}^t u_i v_i$, they can simply consider the state $U^\dagger |\tilde{\psi}_i\rangle$ and output 0 or 1, depending on whether $U^\dagger|\tilde{\psi}_i\rangle$ is closer to $|0\rangle$ or to $|1\rangle$. 
This strategy gives them the correct bit with probability $3/4$. 
Overall, this strategy leads to an expected fraction of correct bits equal to $3/4$, which means that the protocol $G_{\mathrm{IP}}(n,t, 1/4, 0)$ is not secure against non entangled cheaters. 

If $\eta_{\mathrm{loss}}>0$, that is if losses are tolerated, then Alice and Bob can apply the same technique and return a value only if $\max \{ |\langle 0 | U^\dagger | \tilde{\psi}_i\rangle|^2, |\langle 0 | U^\dagger | \tilde{\psi}_i\rangle|^2\}$ is large enough. A similar analysis as in \cite{QS15}  shows that if Alice and Bob only return a value for a fraction $1-\eta_{\mathrm{loss}}$ of the qubits, then their error rate is $(1-\eta_{\mathrm{loss}})/4$. 
This shows that non entangled cheaters have a winning strategy as soon as  $\eta_{\mathrm{err}} + \eta_{\mathrm{loss}}/4 \geq 1/4$.

We leave as an open question whether there exist subexponential strategies allowing the cheaters to win the game with non negligible probability when  $\eta_{\mathrm{err}} + \eta_{\mathrm{loss}}/4 \leq 1/4 - \eps$ for some small $\eps>0$.

\section{Loss-tolerant protocols}

In general, the strategies consisting in measuring the state in a random basis allow the cheaters to win a constant fraction of the $n$ ``rounds'' of a game. This is problematic because it seems that a honest prover cannot do much better as soon as the quantum channel from the verifiers is imperfect, either lossy or noisy. As a consequence, it would appear that position verification is not robust against losses or noise (see \cite{QS15} for possible trade-offs between loss and noise). Fortunately, this conclusion is a little bit too pessimistic. 

For instance, the Interleaved Product protocol can be straightforwardly modified to be made loss-tolerant, provided that the prover has access to a good quantum memory. The crucial point to note here is that this protocol appears to remain secure even if the quantum state is distributed in advance compared to the classical information required to decide in which basis to measure the state or to which verifier it should be forwarded. 
From this observation, we propose the following modification of the Interleaved Product protocol:

\emph{In addition to the verifiers, there is a central ``bank'' of quantum states available to the prover. This bank (whose role can be played by the verifiers) distributes quantum states, along with some identification number, to interested parties. The value of the states is not revealed to the client but the verifiers have access to a complete listing of pairs: (state ID, state value). When a prover wants to authenticate her position thanks to a position verification protocol, she should therefore obtain a quantum state from the bank, put it in a quantum memory, and then inform the verifiers of the state ID. Then, the verifiers can apply the usual protocol, with the exception that the state $|\psi\rangle$ does not need to be distributed since the game is played with the state the prover obtained from the bank.}

It seems to us that this modified protocol remains as secure as the original Interleaved Product protocol. More precisely, we could not think of any attack working against the modified version that would not also work against the original version. 

The advantage of this modified version is that the quantum channel between the verifiers and the prover is replaced by the quantum memory of the prover. This could become quite advantageous in a scenario where the physical distance between the verifiers and the prover is large, meaning that fiber optics communication would lead to high losses, provided that the prover has access to a good quantum memory. While the current state-of-the-art on quantum memories (see for instance \cite{simon10} for a recent review) is certainly not sufficient to implement this modified version of the protocol, there are no reason to doubt that high fidelity quantum memories with long coherence time will not become available in the future.

\section{Discussion \& Conclusion}

In this paper we have first studied a general family of attack strategies against position based quantum cryptography.
In particular, we have established a connection between several well studied quantum information processing tasks and position based quantum cryptography. It was previously known that there exists some efficient attack when the verifiers choose the challenge unitary from Clifford group. Here, we showed that this remains true if the unitaries lie in a low level of the Clifford hierarchy. This result connects notions relevant in fault-tolerant quantum computing with the attack complexity of position based quantum cryptography.

Then, we have introduced a very practical position-verification scheme, the Interleaved Product protocol, which appears to be immune to these attacks and displays the further advantage of being loss-tolerant in a scenario where the quantum state is distributed independently from the classical challenge.

\begin{acknowledgements}
We are particularly grateful to Florian Speelman for informing us of a mistake in Theorem \ref{circ_cliff} in a previous version of this paper. After completion of this work, we learned that Florian Speelman had independently proved Theorem \ref{thm1} \cite{spe15}.

We also thank Andr{\'e} Chailloux, Fr{\'e}d{\'e}ric Grosshans and Christian Schaffner for many stimulating discussions on position-based cryptography. 
\end{acknowledgements}

\appendix

\newpage

\section{Technical tools}
\label{prelim}

In this appendix, we review some technical notion used in the rest of the paper: the Clifford Hierarchy, teleportation gates, semi-Clifford gates and port-based teleportation.

\subsection{The Clifford Hierarchy}
\label{cliffhierar}
The \textit{Clifford Hierarchy} introduced in \cite{GC99} is an infinite hierarchy of sets $C_1(n) \subset C_2 (n)\subset \cdots \subset C_k(n) \cdots $ of $n$-qubit unitaries where $C_1(n) = \mathcal{P}_n$ corresponds to the Pauli group (on $n$ qubits), and the higher levels are defined recursively by:
$$ U \in C_{k+1}(n) \; \text{if and only if} \; U \sigma U^\dagger \in C_k(n) \; \text{for all} \; \sigma \in C_1(n).$$
When $n$ is clear from context, we simply write $C_k$ instead of $C_k(n)$ for the $k^{\mathrm{th}}$ level of the Clifford hierarchy for $n$-qubit gates.
It should be noted that the first two levels of the hierarchy are groups, namely the Pauli and the Clifford groups, whereas none of the higher levels are groups. 

The gates from $C_1$ and $C_2$ can be ``easily'' implemented fault tolerantly \cite{G97}. 
However, it is well known that they do not form a universal set for quantum computation. One therefore requires at least one gate from $C_3$ to obtain a universal set of gates. 
Not surprisingly, gates from $C_3$ or higher levels are usually much harder to implement fault-tolerantly.

\subsection{Teleportation Gates}

\textit{Teleportation gates} are a tool introduced by Gottesman and Chuang \cite{GC99} to implement a unitary operator $U$ on any state provided that one can apply it to a special state. In particular, teleportation and the ability to perform single qubit operators are sufficient to obtain (fault-tolerant) universal quantum computation.

The main idea relies on the fact that if one uses the state $(I\otimes U)|\Phi^+\rangle$ instead of 
$|\Phi^+\rangle = \frac{1}{\sqrt{2}}(|00\rangle +|11\rangle)$ to teleport a quantum state $|\psi\rangle$ then the teleported state will be of the form $U|\psi\rangle$ (up to some Pauli correction).
To implement an $n$-qubit quantum gate $U \in C_3$, one first prepares the state $|\Psi^n_U\rangle = (I \otimes U)|\Phi^+\rangle^{\otimes n}$. Let $|\psi\rangle$ be an unknown state on which $U$ has to be applied. Then taking $|\psi\rangle$ and performing a Bell basis measurement on $|\psi\rangle$ and on the first register of $|\Psi^n_U\rangle$ leaves $n$ qubits in the state $|\psi_{\mathrm{out}}\rangle = UR|\psi\rangle = R^1U|\psi\rangle$, where the correction $R \in C_1$ is a Pauli operator and $R^1 = U RU^\dagger \in C_2$. 
Since $R^1 \in C_2$, its inverse can easily be implemented, thus giving the state $U|\psi\rangle$. 
Hence, using only $n$ EPR pairs, one can implement any $n$-qubit quantum gate from $C_3$ provided that the state $|\Psi^n_U\rangle$ can be prepared efficiently. 

If $U$ belongs to some higher level $C_k$ with $k>3$ of the Clifford hierarchy, then one can apply the technique outlined above iteratively for $k-2$ steps. Indeed, in that case, the correction $R^1$ belongs to $C_{k-1}$. It should be clear that higher levels of the hierarchy require more teleportation steps and Bell measurements.

\subsection{Semi-Clifford Gates}

Semi-Clifford gates are another special type of gates with different structural properties than the gates in Clifford hierarchy. The concept of semi-Clifford gates was first introduced for the single-qubit case by D. Gross and M. Van den Nest in \cite{GV08}, and generalized to $n$-qubit states by Zeng \textit{et al} in \cite{ZCC08}.
\begin{definition}
 An $n$-qubit unitary operation is called \emph{semi-Clifford} if it sends by conjugation at least one maximal abelian subgroup of $\mathcal{P}_n$ to another maximal abelian subgroup of $\mathcal{P}_n$.
\end{definition}
In particular, if $U$ is an $n$-qubit semi-Clifford operation, then there must exist at least  one maximal abelian subgroup $G$ of $\mathcal{P}_n$, such that $UGU^{\dagger}$ is another maximal abelian subgroup of $\mathcal{P}_n$.
While the general structure of the semi-Clifford gates is not yet completely understood for arbitrary $n$, we have a characterization for $n =1, 2$ and a partial characterization for $n=3$.
\begin{theorem}[from \cite{ZCC08}]
\label{thm_sem_cliff1}
The gates in $C_k(1), C_k(2)$ are semi-Clifford for all $k$.
 For $n=3$, all the gates in $C_3(3)$ are semi-Clifford.
\end{theorem}

In our work, semi-Clifford gates will be of interest as they allow the cheaters to perform more efficient attack strategies for the second family of protocols.

\subsection{Port-based teleportation}

Port-based teleportation is a specific teleportation scheme introduced in \cite{IH08}, that allows \textit{Alice} to teleport an arbitrary quantum state to Bob, using many EPR pairs, called \textit{ports}.
After Alice's measurement on her state and her half of the EPR pairs, the state is teleported (approximately) to one of Bob's port, known to Alice. 
Alice simply sends this classical information to Bob, who only needs to trace out the other ports to recover Alice's state. The main feature of this teleportation scheme is that apart from tracing out some registers, Bob needs not apply any correction to the state. 
The fidelity $F_p(|\Psi^{\mathrm{in}}\rangle,|\Psi^{\mathrm{out}}\rangle)$ between Alice's initial state and Bob's final state using port-based teleportation depends on both the number $N$ of EPR pairs consumed in the scheme and the dimension $d$ of Alice's state. 
The following lower-bound was established in \cite{IH09}.
\begin{lemma}[from \cite{IH09}]
\begin{align}
F_p(|\Psi^{\mathrm{in}}\rangle,|\Psi^{\mathrm{out}}\rangle) \geq 1-\frac{d^2}{N}.
\end{align}
\end{lemma}

\end{document}